\documentclass{amsart}

\usepackage{amssymb,latexsym}
\usepackage[latin1]{inputenc} 
\usepackage[dvips]{graphicx,color,psfrag}

\numberwithin{equation}{section}
\numberwithin{figure}{section}
\numberwithin{table}{section}

\theoremstyle{plain}
\newtheorem{prop}{Proposition}[section]
\newtheorem{lemma}[prop]{Lemma}
\newtheorem{theorem}[prop]{Theorem}
\newtheorem{corollary}[prop]{Corollary}

\theoremstyle{definition}

\newenvironment{proofoft}[1]{{\em Proof of Theorem #1. }}{$\Box$ \vspace{1em}}

\newcommand{\E}{{\mathbf E}}

\newcommand{\RR}{{\mathbb R}}

\newcommand{\Remark}{\em Remark: \rm}

\newcommand{\dpt}{{}^{\,\,}\!} %Approx +0.3 pt
\newcommand{\scalefactor}{0.85}
\def\P{{\mathbf P}}

\definecolor{darkred}{rgb}{.7,0,0}

\begin{document}

%Article info
\title[Optimal closing of a pair trade]
{Optimal closing of a pair trade with a model containing jumps}
\date{\today}

%Author info
\author[S.~Larsson]{Stig Larsson$^1$}
\address[S.~Larsson]{Department of Mathematical Sciences, Chalmers University of Technology and University of Gothenburg, SE-41296 Gothenburg, Sweden}
\email[S.~Larsson]{stig@chalmers.se}
\urladdr[S. Larsson]{http://www.math.chalmers.se/~stig}
\thanks{$^1$Supported by the Swedish Research Council (VR) and by the Swedish Foundation for Strategic Research (SSF) through GMMC, the Gothenburg Mathematical Modelling Centre.}

\author[C.~Lindberg]{Carl Lindberg}
\address[C.~Lindberg]{Department of Mathematical Sciences, Chalmers University of Technology and University of Gothenburg, SE-41296 Gothenburg, Sweden}
\email[C.~Lindberg]{clind@chalmers.se}
%\urladdr{http://www.chalmers.se/math/SV/organisation/matematisk-statistik/personal/larare-och-forskare/lindberg-carl}
\thanks{}

\author[M.~Warfheimer]{Marcus Warfheimer$^2$}
\address[M.~Warfheimer]{Department of Mathematical Sciences, Chalmers University of Technology and University of Gothenburg, SE-41296 Gothenburg, Sweden}
\email[M.~Warfheimer]{marcus.warfheimer@gmail.com}
\urladdr[M. Warfheimer]{http://www.math.chalmers.se/~warfheim}
\thanks{$^2$Research partially supported by the Göran Gustafsson Foundation for Research in Natural Sciences and Medicine.}

%AMS info

\begin{abstract}
A pair trade is a portfolio consisting of a long position in one asset and 
a short position in another, and it is a widely applied investment strategy in the financial industry. Recently, Ekstr\"om, Lindberg and Tysk studied the problem of optimally closing a pair trading strategy when the difference of the two assets is modelled by an Ornstein-Uhlenbeck process. In this paper we study the same problem, but the model is generalized to also include jumps. More precisely we assume that the above difference is an Ornstein-Uhlenbeck type process, driven by a L\'evy process of finite activity. We prove a verification theorem and analyze a numerical method for the associated free boundary problem. We prove rigorous error estimates, which are used to draw some conclusions from numerical simulations.
\end{abstract}

\maketitle

\section{Introduction}\label{paper4:intro}

A portfolio which consists of a positive position in one asset, and a negative position in another is called a pair trade. Pairs trading was developed at Morgan Stanley in the late 1980's, and today it is one of the most common investment strategies in the financial industry. The idea behind pairs trading is quite intuitive: the investor finds two assets, for which the prices have moved together historically. When the price spread widens, the investor takes a short position in the outperforming asset, and a long position in the underperforming one with the hope that the spread will converge again, generating a profit. A main advantage of pairs trading is that the short position can, in principle, remove any exposure to market risk. For a historical evaluation of pairs trading we refer to \cite{Gatev}. 

To model the pair spread the authors in \cite{Elliot_Hoek_Malcolm} proposed a mean reverting Gaussian Markov chain which they considered to be observed in Gaussian noise. Recently, in \cite{Ekstrom_Lindberg_Tysk} the authors suggested the continuous time analogue, the so called mean reverting Ornstein-Uhlenbeck process. In this paper we generalize the model of the spread to also include possible jumps. Let $(\Omega,\mathcal F, \P)$ be a complete probability space where the following processes are defined in such a way that they are independent:
\begin{itemize}
\item[--] A standard Brownian motion $W=\{W_t\}_{t\geq 0}$.
\item[--] A Possion process $N^\lambda=\{N_t^\lambda\}_{t \geq 0}$ with intensity $\lambda>0$.
%\item[--] A sequence of independent random variables $\{X_k^\varphi\}_{k=1}^\infty$ with common continuous symmetric density $\varphi$ with $\int_\RR y^2\varphi(y)\,dy<\infty$.
\item[--] A sequence of independent random variables $\{X_k^\varphi\}_{k=1}^\infty$ with common continuous symmetric density $\varphi$. Moreover, the support of $\varphi$ is contained in the interval $(-J,J)$ for some $J>0$.
\end{itemize}
Define the compound Poisson process $C^{\lambda,\varphi}=\{C_t^{\lambda,\varphi}\}_{t\geq 0}$ in the usual way as 
\[
C_t^{\lambda,\varphi}=\displaystyle\sum_{k=1}^{N^\lambda_t} X_k^\varphi
\]
and denote the filtration generated by $W$, $C^{\lambda,\varphi}$ and the null sets of $\mathcal F$ by $\mathbb F=\{\mathcal F_t\}_{t\geq 0}$. It is well known that this filtration satisfies the usual hypotheses (see for example \cite{Protter}). From now on, when we say that a process is a martingale, submartingale or supermartingale we mean that this is with respect to $\mathbb F$. 

Let the difference $U=\{U_t\}_{t\geq 0}$ between the assets be the unique solution of the stochastic differential equation
\begin{equation}\label{paper4:sde_OU}
dU_t = -\mu U_t \,dt+\sigma \,dW_t+ dC_t^{\lambda,\varphi},\quad t>0,
\end{equation}
where $\mu>0$, $\sigma>0$. (The solution of equation \eqref{paper4:sde_OU} is usually called a generalized Ornstein-Uhlenbeck process or an Ornstein-Uhlenbeck type process.) Sometimes we will denote the driving L\'{e}vy process in \eqref{paper4:sde_OU} by $Z^{\sigma,\lambda,\varphi}$, i.e.\
\[
Z^{\sigma,\lambda,\varphi}_t=\sigma W_t+C_t^{\lambda,\varphi},\quad t\geq 0.
\]
As discussed in \cite{Ekstrom_Lindberg_Tysk}, there is a large risk associated with a pair trading strategy. Indeed, if the market spread ceases to be mean reverting, the investor is exposed to substantial risk. Therefore, in practice the investor typically chooses in advance a stop-loss level $a<0$, which corresponds to the level of loss above which the investor will close the pair trade. For a given stop-loss level $a<0$ define 
\begin{equation}\label{paper4:tau_a}
\tau_a=\inf \{t\geq 0:\, U_t\leq a\},
\end{equation}
the first hitting time of the region $(-\infty,a]$, and the so called value function
\begin{equation}\label{paper4:V}
V(x)=\sup_{\tau} \E_x[U_{\tau_a\wedge\tau}]\quad x\in\RR,
\end{equation}
where the supremum is taken over all stopping times with respect to
$U$. (Here and in the sequel $\E_x$ means expected value when
$U_0=x$.) The major interest here is to characterize $V$, and perhaps
more importantly, to describe the stopping time where the supremum is
attained. Since the drift has the opposite sign as $U$, we have no
reason to liquidate our position as long as $U$ is negative. On the
other hand, if $U$ is positive, then the drift is working against the
investor and for large values of $U$ the size of the drift should
overcome the possible benefits from random variations. Moreover, since
the jumps are assumed to be symmetric, this indicates that there is a
stopping barrier $b>0$ with the property that we should keep our
position when $U_t<b$ and liquidate as soon as $U_t\geq b$. We note
that we cannot be sure to close the pair trade at any of the
boundaries $a$ or $b$, because the spread can exhibit jumps. This was not the case in \cite{Ekstrom_Lindberg_Tysk} and it is the major reason for the additional difficulties encountered in the present paper.

General optimal stopping theory (described for example in \cite[Ch.~3]{Peskir_Shiryaev}) %\citep[Ch.~3]{Peskir_Shiryaev}) 
leads us to believe that the value function is given by $V=u$, where $(u,b)$ is the solution of the free boundary problem
\begin{align}\label{paper4:FBP}
\begin{aligned}
\mathcal{G}_U u(x)&=0, && x\in (a,b), \\
u(x)&=x, && x\not\in (a,b), \\
u^\prime(b)&=1. 
\end{aligned}
\end{align}
Here $\mathcal{G}_U$ is the infinitesimal generator of $U$, which is defined on the space of twice continuously differentiable functions $f:\RR\to\RR$ with compact support:
\begin{equation}\label{paper4:gen_U}
\mathcal{G}_U f(x)=\frac{\sigma^2}{2} f^{\prime\prime}(x)-\mu x f^\prime(x)+\lambda\int_{-\infty}^\infty (f(x+y)-f(x))\varphi(y)\,dy, \quad x\in\RR.
\end{equation}
Moreover, the stopping time where the supremum in \eqref{paper4:V} is attained should be 
\begin{equation}\label{paper4:tau_b}
\tau_b=\inf\{t\geq 0:\,U_t\geq b\}.
\end{equation}
Indeed, our first result is a so called verification theorem.
\begin{theorem}\label{paper4:ver_thm}
Assume that $(u,b)$ is a classical solution of \eqref{paper4:FBP} with
\begin{itemize}
\item[a)] $\mathcal{G}_U u(x)\leq 0$, for $x>b$,
\item[b)] $u(x)\geq x$, for $x\in\RR$.
\end{itemize}
Then $u(x)=V(x)=\E_x[U_{\tau_a\wedge\tau_b}]$, for $x\in\RR$, where $V$ is given by \eqref{paper4:V}.       
\end{theorem}
\noindent\Remark As seen from the assumptions on $\varphi$, we are assuming that the absolute value of the jumps of the process $\{U_t\}_{\geq 0}$ are bounded. The reason is that on the financial market, an asset cannot jump to arbitrarily large levels. If nothing else, the jumps are bounded by all the money in the world.

The rest of the paper is organized as follows. In Section \ref{paper4:proof_ver_thm} we prove Theorem \ref{paper4:ver_thm} and in Section \ref{paper4:numerical_FBP} we discuss a numerical solution of the free boundary problem \eqref{paper4:FBP}. We also present strong evidence for the existence and uniqueness of a solution of \eqref{paper4:FBP}.

\section{Proof of Theorem \ref{paper4:ver_thm}}\label{paper4:proof_ver_thm}

Before we start to prove Theorem~\ref{paper4:ver_thm} we need to recall some facts. From the general theory in \cite{Garroni_Menaldi} we get that the boundary value problem 
\begin{align}\label{paper4:BVP}
\begin{aligned}
\mathcal{G}_U u(x)&=0, && x\in (a,b), \\
u(x)&=x, && x\not\in (a,b),
\end{aligned}
\end{align}
has a unique classical solution and that such a solution belongs to the space
\[
C^2(\RR\setminus\{a,b\})\cap C^1(\RR\setminus\{a,b\})\cap C(\RR).
\]
Moreover, the finite left and right limits of $u^\prime$ and $u^{\prime\prime}$ exist at $a$ and $b$. Although these facts follow from \cite{Garroni_Menaldi}, we present in Theorem~\ref{paper4:thm:existence} a self-contained proof for the simpler situation that we consider here. Hence, if $(u,b)$ is a classical solution of \eqref{paper4:FBP}, then necessarily 
\[
u\in C^2(\RR\setminus\{a,b\})\cap C^1(\RR\setminus\{a\})\cap C(\RR)
\]
with finite left and right limits of $u^\prime$ and $u^{\prime\prime}$ everywhere. Furthermore, recall a generalized version of It\^o's formula for convex functions (see for example \cite[Ch.~4]{Protter}):
\begin{theorem}[Meyer-It\^o formula]\label{paper4:meyer_ito}
Let $X=\{X_t\}_{\geq 0}$ be a semimartingale and let $f$ be the difference of two convex functions. Then 
\[
\begin{split}
f(X_t)&=f(X_0)+\int_{0+}^t D^-f(X_{s-}) \,dX_s \\
&\quad+ \displaystyle\sum_{0<s\leq t} \big(f(X_s)-f(X_{s-})-D^-f(X_{s-})\dpt\Delta X_s\big) \\
&\quad+\frac{1}{2}\int_{-\infty}^\infty L_t^y(X)\, d\mu(y),
\end{split}
\]
where $D^-f$ is the left derivative of $f$, $\mu$ is a signed measure which is the second generalized derivative of $f$ and $\{L_t^a(X)\}_{t\geq 0}$ is the local time process of $X$ at $a$. 
\end{theorem}
Due to the regularity of $u$ it can be written as a difference of two convex functions (see Problem $6.24$ in \cite[Ch.~3]{Kar_Shr}). Moreover, the second derivative measure $\mu$ of $u$ can be split into two parts $\mu=\mu_c+\mu_d$, where the continuous part $\mu_c$ is given by $d\mu_c=u^{\prime\prime}\,dx$ and the discrete part $\mu_d=\delta_a$ is a point mass at $a$. Here, $u^{\prime\prime}(x)$ denotes the second derivative of $u$ at $x$ except at the points $a$ and $b$, where it denotes the right second derivative (which we know is finite). By Corollary $1$ of the Meyer-It\^o formula in \cite{Protter}, we can now write
\begin{equation}\label{paper4:localtime}
\begin{split}
\frac{1}{2}\int_{-\infty}^\infty L_t^y(U) \,d\mu(y)&=\frac{1}{2}\int_0^t u^{\prime\prime}(U_{s-}) \,d[U,U]_s^c+\frac{1}{2} L_t^a(U)\big(u^{\prime}(a+)-u^{\prime}(a-)\big) \\
&=\frac{\sigma^2}{2}\int_0^t u^{\prime\prime}(U_{s-}) \,ds+\frac{1}{2} L_t^a(U)\big(u^{\prime}(a+)-u^{\prime}(a-)\big),
\end{split}
\end{equation}  
where $[U,U]^c$ denotes the continuous part of the quadratic variation $[U,U]$.

Furthermore, by using \eqref{paper4:sde_OU} and the compensated Poisson random measure
\[ 
\tilde{N}_Z(dt,dy)=N_Z(dt,dy)-\lambda \,dt\,\varphi(y) \,dy, 
\]
where $N_Z$ denotes the jump measure associated with $Z^{\sigma,\lambda,\varphi}$,  we get 
\begin{equation}\label{paper4:dX_s}
\begin{split}
\int_{0+}^t D^-u(&U_{s-})\,dU_s+\displaystyle\sum_{0<s\leq t} \big(u(U_s)-u(U_{s-})-D^-u(U_{s-})\dpt\Delta U_s\big)\\
&=-\mu\int_0^t U_{s-}D^-u(U_{s-})\,ds+\sigma\int_0^t D^-u(U_{s-})\,dW_s \\
&\quad+\int_{0+}^t\int_\RR\big(u(U_{s-}+y)-u(U_{s-})\big) \,\tilde{N}_Z(ds,dy) \\
&\quad+\lambda\int_0^t\int_\RR\big(u(U_{s-}+y)-u(U_{s-})\big)\varphi(y)\,dy\,ds.
\end{split}
\end{equation}
Summing up, we now have for $t\geq 0$
\begin{equation}\label{paper4:ito_u}
\begin{split}
u(U_t)&=u(U_0)+\int_0^t \Big(\frac{\sigma^2}{2} u^{\prime\prime}(U_{s-})-\mu U_{s-}D^-u(U_{s-})\Big)\,ds \\
&\quad+\lambda\int_0^t\int_\RR \big(u(U_{s-}+y)-u(U_{s-})\big)\varphi(y)\,dy\,ds \\
&\quad+\frac{1}{2} L_t^a(U)\big(u^{\prime}(a+)-u^{\prime}(a-)\big)+M_t,
\end{split}
\end{equation}
where
\[
M_t=\sigma\int_0^t D^-u(U_{s-})\,dW_s+\int_{0+}^t\int_\RR\big(u(U_{s-}+y)-u(U_{s-})\big) \,\tilde{N}_Z(ds,dy).
\]
Since $u$ is Lipschitz, has a bounded left derivative and since the jumps density $\varphi$has a finite swe get that $\{M_t\}_{t \geq 0}$ is a martingale.
\begin{lemma}\label{paper4:localtime2}
Assume $a\in\RR$ and $U_0>a$. Then a.s.\ $L_{\tau_a\wedge t}^a(U)=0$ for all $t\geq 0$.
\end{lemma}
\begin{proof}
Fix $a\in\RR$ and assume $U_0>a$. Since the local time process $\{L_t^a\}_{t\geq 0}$ is continuous in $t$ it is enough to prove that for fixed $t\geq 0$ we have $L_{\tau_a\wedge t}^a(U)=0$ a.s. From \cite[p.~217]{Protter}, we get that
\[
\begin{split}
\frac{1}{2}L_{\tau_a\wedge t}^a(U)&=(U_{\tau_a\wedge t}-a)^{-}-\displaystyle\sum_{0<s\leq\tau_a\wedge t}1_{\{U_{s-}>a\}}(U_s-a)^- \\
&\quad+\int_{0+}^{\tau_a\wedge t}1_{\{U_{s-}\leq a\}}\,dU_s-\displaystyle\sum_{0<s\leq\tau_a\wedge t}1_{\{U_{s-}\leq a\}}(U_s-a)^+.
\end{split}
\]
Futhermore, from the fact that $U_{s}> a$ for all $0<s< \tau_a\wedge t$, we get that $U_{s-} \geq a$ for all $0<s< \tau_a\wedge t$ and from the left continuity of $U_{s-}$, we can conclude that we also have $U_{\tau_a\wedge t-} \geq a$. From that and by splitting the integral and the sum, we obtain
\[
\begin{split}
\frac{1}{2}L_{\tau_a\wedge t}^a(U)&=1_{\{U_{\tau_a\wedge t-}=a\}}(U_{\tau_a\wedge t}-a)^-+1_{\{U_{\tau_a\wedge t-}=a\}}(U_{\tau_a\wedge t}-a)\\
&\quad-1_{\{U_{\tau_a\wedge t-}=a\}}(U_{\tau_a\wedge t}-a)^++\int_{0+}^{\tau_a\wedge t-}1_{\{U_{s-}=a\}}\,dU_s \\
&\quad-\displaystyle\sum_{0<s<\tau_a\wedge t}1_{\{U_{s-}=a\}}(U_s-a)^+ \\
&=\int_{0+}^{\tau_a\wedge t-}1_{\{U_{s-}=a\}}\,dU_s-\displaystyle\sum_{0<s<\tau_a\wedge t}1_{\{U_{s-}=a\}}(U_s-a)^+.
\end{split}
\]
From the observation that if $U_{s-}=a$ for some $0<s<\tau_a\wedge t$, then $s$ is a jump time and the jump must be in the up direction, we conclude that the right hand side of the last expression is zero and so we are done.  
\end{proof} 

\noindent
\Remark In a similar way one can show that, if $a<U_0<b$, then 
\[
L_{\tau_a\wedge\tau_b\wedge t}^a(U)= 0\;\, \text{and}\;\, L_{\tau_a\wedge\tau_b\wedge t}^b(U)= 0 \;\,\text{for}\;\, t\geq 0. 
\]

\noindent
\begin{proofoft}{\ref{paper4:ver_thm}}
Since $u(x)=V(x)=\E_x[U_{\tau_a\wedge\tau_b}]=x$, when $x\leq a$, we can assume that $x>a$. Define $Y_t=u(U_{\tau_a\wedge t})$, $t\geq 0$. By using \eqref{paper4:ito_u}, Lemma \ref{paper4:localtime2}, the expression \eqref{paper4:gen_U} for the generator of $U$, and \eqref{paper4:FBP}, we get
\begin{equation}\label{paper4:Y}
\begin{split}
Y_t&=u(x)-\int_0^{\tau_a\wedge t} \mu U_{s-}1_{\{U_{s-}\geq b\}} \,ds \\
&\quad+\lambda\int_0^{\tau_a\wedge t}\int_\RR\big(u(U_{s-}+y)-u(U_{s-})\big)\varphi(y)1_{\{U_{s-}\geq b\}}\,dy\,ds+M_{\tau_a\wedge t}.
\end{split}
\end{equation}
Property $a)$ and the martingale property of $\{M_{\tau_a\wedge t}\}$ give that $\{Y_{t}\}_{t \geq 0}$ is a supermartingale. Furthermore, from property $b)$ we get that $Y_{t}\geq U_{\tau_a\wedge t}$, for $t\geq 0$, and since
\begin{equation}\label{paper4:ineq_U}
U_{\tau_a\wedge t}\geq a-J,\quad t\geq 0,
\end{equation} 
we can apply the optional sampling theorem (see \cite{Kar_Shr}) and obtain
\[
\E_x[U_{\tau_a\wedge\tau}]\leq \E_x[Y_{\tau}]\leq \E_x[Y_0]=u(x),
\]
where $\tau$ is an arbitrary stopping time with respect to $U$. Hence, $V(x)\leq u(x)$ for $x>a$. % By applying Fatou's lemma to $U_{\tau_a\wedge\tau_n\wedge\tau}-H$ we can conclude that $\E_x[U_{\tau_a\wedge\tau}]\leq u(x)$ and hence $V(x)\leq u(x)$ for $x>a$.
In particular, if $x\geq b$ then $x\leq V(x)\leq u(x)=x$ and so $u(x)=$ $V(x)=$ $\E_x[U_{\tau_a\wedge\tau_b}]$ when $x\geq b$.

For the case when $a<x<b$, note that from \eqref{paper4:Y} we get for $t\geq 0$ that
\[
Y_{\tau_b\wedge t}=M_{\tau_a\wedge\tau_b\wedge t}+u(x)
\]
and since 
\[
a-J\leq Y_{\tau_b\wedge t}\leq b+J,\quad t\geq 0,
\]
the optional sampling theorem applies again and we obtain $u(x)=\E_x[Y_{\tau_b}]$. Finally, the fact that $Y_{\tau_b}=U_{\tau_a\wedge\tau_b}$ gives us $u(x)=\E_x[U_{\tau_a\wedge\tau_b}]\leq V(x)$ and the proof is complete.
\end{proofoft}

\section{Numerical solution of the free boundary value
  problem}\label{paper4:numerical_FBP}

We have not been able to give a rigorous proof of the existence and
uniqueness of the solution $(u,b)$ of the free boundary value problem
\eqref{paper4:FBP}. We therefore resort to a numerical solution by
means of the finite element method. However, at the end of this
section we will show that we have strong computational evidence for
both existence and uniqueness for \eqref{paper4:FBP}.  In order to
  achieve this we first show rigorous existence and regularity results
  for the boundary value problem \eqref{paper4:BVP} and rigorous convergence estimates
  with explicit constants for the finite element approximation.
 
\subsection{The boundary value problem}  \label{paper4:subsec:1} 
We begin by transforming the free boundary value problem
\eqref{paper4:FBP} to a problem with homogeneous boundary values. Set
$v(x)=u(x)-x$ and use 
$
\int_{-\infty}^\infty y\varphi(y)\,dy=0
$
to get 
\begin{align}\label{paper4:FBP2}
\begin{aligned}
 -\tfrac12{\sigma^2} v''(x)+\mu xv'(x) \phantom{+y)-v(x)\big)\varphi(y)\,dy}& &&\\
 -\lambda \int_{-\infty}^\infty
    \big(v(x+y)-v(x)\big)\varphi(y)\,dy
        &=-\mu x,    && x\in(a,b), \\ 
   v(x) &=0,      && x\not\in (a,b),\\ 
   v'(b)&=0.  
\end{aligned}
\end{align}
Introducing the operators 
\begin{align*}
 &\mathcal{L}v(x)= -\tfrac12{\sigma^2} v''(x)+\mu xv'(x),  \\
  &\mathcal{I}v(x)= \lambda \int_{-\infty}^\infty
 \big(v(x+y)-v(x)\big)\varphi(y)\,dy, 
\end{align*}
our approach will be to first solve the boundary value problem
\begin{align}
 \label{paper4:eq:3}
 \begin{aligned}
  \mathcal{L}v 
   -\mathcal{I}v 
   &=f, \ &&  x\in(a,b), \\ 
  v(x)&=0, && x\not\in (a,b),
\end{aligned}
\end{align}
with $f(x)=-\mu x$, and then for fixed $a<0$ find $b>a$ such that
$v'(b)=0$. 

To solve \eqref{paper4:eq:3} we follow a standard approach based on a
weak formulation and Fredholm's alternative.  We denote by
$(\cdot,\cdot)$ and $\|\cdot\|$ the standard scalar product and norm
in $L_2(a,b)$, and we denote by $H^k(a,b)$ and $H^1_0(a,b)=\{v\in
H^1(a,b):v(a)=v(b)=0\}$ the standard Sobolev spaces. We denote the
derivative $Dv=dv/dx$. We choose $v\mapsto\|Dv\|$ to be the norm in
$H^1_0(a,b)$, which is equivalent to the standard $H^1$-norm.  We
extend functions $v\in L_2(a,b)$ by zero outside $(a,b)$ so that
$\mathcal{I}v$ is properly defined. We define bilinear forms
\begin{align}
  \label{paper4:eq:4}
  \begin{split}
 A_{\mathcal{L}}(u,v)
 &=\int_a^b \big(\tfrac12\sigma^2 u'(x)v'(x)+\mu xu'(x)v(x) \big)\,dx,
 \quad u,v\in H^1_0(a,b),\\
 A_{\mathcal{I}}(u,v)
 &=\int_a^b \mathcal{I}u(x)v(x)\,dx, 
\quad u,v\in L_2(a,b)\\
 A(u,v)
 &=  A_{\mathcal{L}}(u,v)-A_{\mathcal{I}}(u,v).  
\end{split}
\end{align}
Since $\int_{-\infty}^\infty \varphi(y)\,dy=1$,
$\varphi(-y)=\varphi(y)$, and $v(x)=0$ for $x\not\in(a,b)$, we also
have
\begin{align}
  \label{paper4:eq:8}
  \mathcal{I}v(x)
  = \lambda \int_{a}^b \varphi(x-y) v(y)\,dy-\lambda v(x) , 
  \quad v\in L_2(a,b). 
\end{align}

%It is important to compute various constants explicitly in terms
%of the parameters in \eqref{paper4:FBP}.  In doing so we assume, without
%loss of generality, that $a<0$, $b>0$. 

The convolution operator $\mathcal{I}_1v(x)=\int_{-\infty}^\infty
\varphi(x-y) v(y)\,dy$ is bounded in $L_2(a,b)$ with constant $
c=\int_{-\infty}^\infty\varphi(y)\,dy=1$ by Young's inequality.
Hence, 
\begin{align}
 \label{paper4:eq:9}
\|\mathcal{I}v\|& \le 2 \lambda\|v\| , \quad  v\in L_2(a,b),\\
 \label{paper4:eq:9b}
\|D\mathcal{I}v\|& \le 2 \lambda\|Dv\| , \quad  v\in H^1_0(a,b),
\end{align} 
and
\begin{align*}
  -A_{\mathcal{I}}(v,v) 
  \ge \lambda\big(\|v\|^2-\|\mathcal{I}_1v\|\|v\|\big)\ge0,  
  \quad v\in L_2(a,b).
\end{align*}
Hence, 
\begin{align*}
\begin{split}
|A(u,v)| 
&\le \tfrac12\sigma^2 \|Du\|\|Dv\|
+\mu\max(|a|,|b|) \|Du\|\|v\|
+2\lambda\|u\|\|v\| \\
&\le c_1 \|Du\|\|Dv\|, \quad  u,v\in H^1_0(a,b), \\
c_1&=\tfrac12\sigma^2+c_2(\mu\max(|a|,|b|)+2\lambda c_2), 
\end{split}
\end{align*}
where we also used Poincar\'e's inequality 
\begin{align}  \label{paper4:poincare}
 \|v\|\le c_2\|Dv\|, \quad v\in H^1_0(a,b), \quad c_2=(b-a)/\pi. 
\end{align}
By integration by parts we obtain
\begin{align*}
A_{\mathcal{L}}(v,v)
=\tfrac12\sigma^2 \|Dv\|^2-\tfrac12\mu \|v\|^2 , 
\quad  v\in H^1_0(a,b),
\end{align*}
so that $A(\cdot,\cdot)$ is bounded and coercive on $H^1_0(a,b)$:
\begin{align}
\label{paper4:eq:12b}
|A(u,v)| 
&\le c_1 \|Du\|\|Dv\|, 
&  u,v\in H^1_0(a,b) , \\ 
\label{paper4:eq:12}
A(v,v)
&\ge\tfrac12{\sigma^2} \|Dv\|^2-\tfrac12{\mu} \|v\|^2 , 
& v\in H^1_0(a,b). 
\end{align}

We say that $v\in H^1_0(a,b) $ is a weak solution of
\eqref{paper4:eq:3} if
\begin{align}
 \label{paper4:eq:5}
  A(v,\phi)=(f,\phi) \quad \forall \phi\in H^1_0(a,b).  
\end{align}  
We also use the adjoint
problem: find $w\in H^1_0(a,b)$ such that
\begin{align}
 \label{paper4:eq:5a}
  A(\phi,w)=(\phi,g) \quad \forall \phi\in H^1_0(a,b).  
\end{align}
The strong form is (note that $\mathcal{I}$ is self-adjoint in $L_2(a,b)$)  
\begin{align}  \label{paper4:eq:5bb}
 \begin{aligned}
   \mathcal{L}^*w(x)-\mathcal{I}w(x)
    &=g(x), \ &&  x\in(a,b), \\ 
   w(x)&=0, && x\not\in (a,b), 
 \end{aligned}
\end{align}
where 
\[
\mathcal{L}^*w(x)=-\tfrac12{\sigma^2} w''(x)-\mu xw'(x) -\mu w(x).
\]

We may now prove the existence and uniqueness of a classical solution
of \eqref{paper4:eq:3}.  In principle this follows from the general
theory in \cite{Garroni_Menaldi}, but we present a self-contained
proof, with explicit constants, for the simpler situation that we
consider here.  The theorem also provides results necessary for the
analysis of the finite element method.

\begin{theorem} \label{paper4:thm:existence} The boundary value
  problem \eqref{paper4:eq:3} has a unique weak solution $v\in
  H^1_0(a,b) $ for every $f\in L_2(a,b)$. The solution belongs to
  $H^2(a,b)$ and there is a constant $c_3$ such that
  \begin{align}
    \label{paper4:eq:26}
    \|D^2v\|\le c_3\|f\|.  
  \end{align}
  Moreover, if $f(x)=-\mu x$, then the solution is classical, $v\in
  C^2([a,b])$.  Similarly, the adjoint problem \eqref{paper4:eq:5bb}
  has a unique weak solution $w\in H^1_0(a,b)$ for each $g\in L_2(a,b)$,
  which belongs to $H^2(a,b)$ and
\begin{align}
 \label{paper4:eq:6b}
 \|D^2w\|\le c_3 \|g\|.  
\end{align} 
\end{theorem}

\begin{proof} The proof is a standard argument as presented, for
  example, in \cite[Ch.~6]{Evans} for elliptic PDEs.  The only
  difference is that that the lowest order term in $A(\cdot,\cdot)$ is
  defined by means of an integral operator, but the crucial
  properties \eqref{paper4:eq:12b}, \eqref{paper4:eq:12} are the
  same.  

  We first show that weak solutions are regular.  We use a regularity
  result for elliptic problems (see \cite[p.~323]{Evans}): If $v$ is a
  weak solution of
\begin{align*}
 \mathcal{L}v(x) =g(x), \ x\in(a,b);  \quad  v(a)=v(b)=0,   
\end{align*}
and if $g\in H^k(a,b)$ for some $k\geq 0$, then $v\in H^{k+2}(a,b)$.
A weak solution $v\in H^1_0(a,b) $ of \eqref{paper4:eq:3} satisfies
this with $g=f+\mathcal{I}v$, where by \eqref{paper4:eq:9}, \eqref{paper4:eq:9b}
$\mathcal{I}v\in H^1(a,b)$.  For $f\in L_2(a,b)$ we conclude that
$v\in H^2(a,b)$.  If $f\in H^1(a,b)$, then we have $v\in H^3(a,b)$ and
by Sobolev's inbedding $v\in C^2([a,b])$. In particular, a weak
solution is classical when $f(x)=0$ and $f(x)=-\mu x$.  Analogous
regularity results hold for the adjoint problem.

Now we can prove existence.  Let
\begin{align*}
   A_{\mu}(u,v) = A(u,v) +\tfrac12\mu (u,v).  
\end{align*}
By the Lax-Milgram lemma we know that the shifted problem 
\begin{align*}
  A_{\mu}(u,\phi) = (g,\phi) 
  \quad\forall \phi\in H^1_0(a,b), 
\end{align*} 
has a unique solution $u\in H^1_0(a,b)$ for each $g\in L_2(a,b)$.  This
defines the bounded linear operator
$\mathcal{A}_{\mu}^{-1}:L_2(a,b)\to H^1_0(a,b)$ by
$u=\mathcal{A}_{\mu}^{-1}g$.  The equation \eqref{paper4:eq:5} is now
equivalent to
\begin{align*}
  v=\mathcal{A}_{\mu}^{-1}f+\tfrac12\mu\mathcal{A}_{\mu}^{-1}v, 
\end{align*}
or $v-Kv=h$, where $h=\mathcal{A}_{\mu}^{-1}f$ and where
$K=\tfrac12\mu\mathcal{A}_{\mu}^{-1}:L_2(a,b) \to L_2(a,b)$ is a
compact operator, because $H_0^1(a,b)$ is compactly inbedded in
$L_2(a,b)$.

By the Fredholm alternative we know that the latter equation is
uniquely solvable for every $h\in L_2(a,b)$ if and only if the
corresponding homogeneous equation has no non-trivial solution.  But a
non-trivial solution of $v-Kv=0$ would be a weak solution, and hence a
classical solution, of \eqref{paper4:eq:3} with $f=0$.

Then we can apply the maximum principle for classical solutions of
\eqref{paper4:eq:3}, see \cite[Theorem 3.1.3]{Garroni_Menaldi}.  It
says that if a classical function satisfies
$(\mathcal{L}-\mathcal{I})u\le0$ in $(a,b)$, then
$\max_{[a,b]}u=\max_{\mathbb{R}\setminus(a,b)}u$.  (The maximum
principle for the integro-differential equation is proved in the same
way as for the differential equation after noting that
$-\mathcal{I}u(x_0)\ge0$ if $u$ has a maximum at $x_0$.) We conclude
that that the homogeneous equation has no non-trivial solution and
therefore \eqref{paper4:eq:3} has a unique weak solution for every
$f\in L_2(a,b)$.  By the Fredholm theory the adjoint problem
\eqref{paper4:eq:5bb} is then also uniquely solvable for all $g\in
L_2(a,b)$.

Finally, we prove the bounds \eqref{paper4:eq:26} and
\eqref{paper4:eq:6b}.  Let $v=\mathcal{A}^{-1}f$ and
$w=(\mathcal{A}^*)^{-1}g$ denote the solution operators of
\eqref{paper4:eq:3} and \eqref{paper4:eq:5bb}, respectively.

Let $f\in H^1_0(a,b)$.  Then $v=\mathcal{A}^{-1}f$ is classical and
the maximum principle gives 
\begin{align}\label{paper4:c_4}
  \|v\|_{L_\infty(a,b)} \le c_4 \|f\|_{L_\infty(a,b)}. 
\end{align}
In order to compute the explicit constant we briefly recall the proof.  Let
\begin{align*}
\phi(x)=
\begin{cases}
e^{\gamma(b-a)}-e^{\gamma (x-a)},     &x\le b, \\
0,     &x\ge b,
\end{cases}
\end{align*}
where $\gamma>0$ is
chosen so that that $\mathcal{A}\phi\ge1$ in $(a,b)$. Then
$u(x)=\|f\|_{L_\infty(a,b)}\phi(x)$ satisfies $\mathcal{A}u\ge
\|f\|_{L_\infty(a,b)}\ge f=\mathcal{A}v $ in $(a,b)$ and $u\ge0=v$ outside
$(a,b)$, so that the maximum principle gives
$\max_{[a,b]}(v-u)=\max_{\mathbb{R}\setminus(a,b)}(v-u)=0$, that is, 
$u\ge v$ in $[a,b]$. Hence $v\le
\|\phi\|_{L_\infty(a,b)}\|f\|_{L_\infty(a,b)}$ in $[a,b]$.  The lower bound $v\ge
-\|\phi\|_{L_\infty(a,b)}\|f\|_{L_\infty(a,b)}$ is obtained in a similar way and so we get
\[
\|v\|_{L_\infty(a,b)} \leq \|\phi\|_{L_\infty(a,b)}\|f\|_{L_\infty(a,b)}\leq e^{\gamma (b-a)} \|f\|_{L_\infty(a,b)}.
\]
To determine $\gamma$, let $x\in (a,b)$ and compute
\[
\begin{split}
-\mathcal{I}\phi(x)&=\lambda e^{\gamma(x-a)}\int_{-\infty}^{b-x} (e^{\gamma y}-1)\varphi(y)\,dy \\
&\quad+\lambda (e^{\gamma(b-a)}-e^{\gamma (x-a)}) \int_{b-x}^\infty \varphi(y)\,dy \\
&\geq 
-\lambda e^{\gamma (x-a)} \int_{-\infty}^\infty \varphi(y)\,dy 
=-\lambda e^{\gamma (x-a)}. 
\end{split}
\]
Hence,
\begin{align*}
  \mathcal{A}\phi(x)
  \ge(\tfrac12 \sigma^2\gamma^2 
  -\mu b \gamma -\lambda) e^{\gamma
  (x-a)}\ge 1 , \quad x\in(a,b), 
\end{align*}
if $\tfrac12 \sigma^2\gamma^2 -\mu b \gamma -\lambda\ge 1$, that
is, if
\[
\gamma=\hat\gamma=\frac{\mu b}{\sigma^2}+\sqrt\frac{2(\lambda+1)}{\sigma^2} .
\]
Then we conclude that \eqref{paper4:c_4} holds with
$c_4=e^{\hat\gamma(b-a)}$. 

Hence, since $\|v\|\le (b-a)^{\frac12}\|v\|_{L_\infty(a,b)} $ and
$\|f\|_{L_\infty(a,b)} \le (b-a)^{\frac12}\|Df\|$, we obtain the bound
\begin{align*}
 \|v\|=\| \mathcal{A}^{-1}f\|\le c_5 \|Df\|  
  \quad \forall f\in H^1_0(a,b), \ c_5=(b-a)c_4.
\end{align*} 
By duality we conclude 
\begin{align*}
 \|(\mathcal{A}^{-1})^*\|_{B(L_2,H^{-1})} =
 \|\mathcal{A}^{-1}\|_{B(H^1_0,L_2)} \le c_5.  
\end{align*}
Hence
\begin{align}\label{paper4:ett}
 \|w\|_{H^{-1}}= \| (\mathcal{A}^*)^{-1}g\|_{H^{-1}} 
  =\|(\mathcal{A}^{-1})^*g\|_{H^{-1}}  \le c_5 \|g\| 
  \quad \forall g\in L_2(a,b), 
\end{align}
where $H^{-1}(a,b)=(H^1_0(a,b))^*$ and 
\begin{align*}  
  \|w\|_{H^{-1}} =\sup_{\phi\in H^1_0}\frac{(\phi,w)}{\|D\phi\|}.
\end{align*}
Recall that $v\mapsto\|Dv\|$ is the chosen norm in $H^1_0(a,b)$.
By using $\phi=w\in H^1_0(a,b)$ here we obtain
\begin{align}
  \label{paper4:tva}
    \|w\|^2\le \|w\|_{H^{-1}} \|Dw\|. 
\end{align}
We take $\phi=w$ in the adjoint equation \eqref{paper4:eq:5a} and use
coercivity \eqref{paper4:eq:12}, the inequality $2ab\le \epsilon
a^2+\epsilon^{-1}b^2$, and \eqref{paper4:tva} to get  
\begin{align*}
 \tfrac12\sigma^2 \|Dw\|^2  
&\le A(w,w)+\tfrac12\mu\|w\|^2 
\le  \|g\|\|w\| + \tfrac12\mu\|w\|^2 
\\&
\le \tfrac12 \mu^{-1}\|g\|^2 + \mu\|w\|^2 
\le  \tfrac12 \mu^{-1} \|g\|^2 + \mu\|w\|_{H^{-1}}\|Dw\|
\\&
\le  \tfrac12 \mu^{-1} \|g\|^2 
+ \mu^2\sigma^{-2} \|w\|_{H^{-1}}^2 +\tfrac14\sigma^2\|Dw\|^2 . 
\end{align*}
With \eqref{paper4:ett} this leads to
\begin{align*}
 \|Dw\|^2
&\le  2\sigma^{-2}\mu^{-1} \|g\|^2 
  +  4\sigma^{-4}\mu^{-2} \|w\|_{H^{-1}}^2 \\
&\le ( 2\sigma^{-2}\mu^{-1}+4\sigma^{-4}\mu^{-2}c_5^2 ) \|g\|^2 
\end{align*} 
and with Poincar\'e's inequality  \eqref{paper4:poincare}, 
\begin{align*}
 \|w\|\le c_2 \|Dw\|
  \le c_2 ( 2\sigma^{-2}\mu^{-1}+4\sigma^{-4}\mu^{-2}c_5^2)^{\frac12}\|g\| .
\end{align*} 
Hence 
\begin{align}  \label{paper4:tre}
  \begin{split}
\| (\mathcal{A}^*)^{-1}g\| &=\|w\|\le c_6 \|g\|  
\quad\forall g\in L_2(a,b), \\ 
c_6&=c_2 ( 2\sigma^{-2}\mu^{-1}+4\sigma^{-4}\mu^{-2}c_5^2)^{\frac12}.  
\end{split}
\end{align}
By duality in $L_2$ we also have
\begin{align}\label{paper4:fyra}
  \|v\| = \| \mathcal{A}^{-1}f\| \le c_6 \|f\| 
  \quad\forall f\in L_2(a,b).
\end{align}
In order to bound $D^2v$ we recall that $v\in H^2(a,b)$. Hence it satisfies
\eqref{paper4:eq:3} strongly, so that with \eqref{paper4:eq:9} we obtain
\begin{align*}
  \tfrac12{\sigma^2} \|D^2v\| 
  &\le \mu \|xDv\| + \|\mathcal{I}v\|+\|f\| \\
  &\le \mu \max(|a|,|b|)\|Dv\| + 2\lambda\|v\|+\|f\| \\
  &\le \mu \max(|a|,|b|)\|D^2v\|^{\frac12}\|v\|^{\frac12} + 2\lambda\|v\|+\|f\| \\
  &\le \tfrac14\sigma^2\|D^2v\| 
  + (2\lambda+\sigma^{-2} \mu^2 \max(|a|,|b|)^2)\|v\|+\|f\|.
\end{align*}
Hence, 
\begin{align*}
  \|D^2v\| &\le c_7 \|f\| + c_8\|v\|, \\
   c_7&=4\sigma^{-2}, \ c_8=4\sigma^{-2} (2\lambda+\mu+\sigma^{-2} \mu^2 \max(|a|,|b|)^2).  
\end{align*}
In the last step we replaced $2\lambda$ by $2\lambda+\mu$ in $c_8$, so
that the same result holds also for the adjoint equation
\eqref{paper4:eq:5bb}.  Using also \eqref{paper4:tre} and
\eqref{paper4:fyra} we finally conclude
\begin{align*}
   \|D^2v\|&\le c_3 \|f\|, \quad 
   \|D^2w\|\le c_3 \|g\|, \\ 
  c_3&=c_7+c_6c_8.   
\end{align*}
This completes the proof.  
\end{proof}

\subsection{The finite element method}  \label{paper4:subsec:2} 
The finite element method is based on a family of subdivisions
$\mathcal{T}_h$ of the interval $[a,b]$ parametrized by the maximal
mesh size $h$. Each mesh is of the form 
\begin{align*}
  \mathcal{T}_h:  a=x_0<x_1<\dots<x_{j-1}<x_{j}<\dots<x_N=b,  \quad
  h=\max_{j=1,\dots,N} (x_{j}-x_{j-1}).  
\end{align*}
We introduce the space $V_h\subset H^1_0(a,b)$
consisting of all continuous functions that reduce to piecewise
polynomials of degree $\le1$ with respect to $\mathcal{T}_h$.  See
\cite[Ch.~5]{LarssonThomee} or \cite[Ch.~1]{BrennerScott}.  Then there
is an interpolator $I_h:C([a,b])\to V_h$ such that $I_hu(x_j)=u(x_j)$,
$j=1,\dots,N$, and 
\begin{align}
  \label{paper4:eq:12a}
  \|D(u-I_hu)\|_{L_p(a,b)}\le h^{\frac12+\frac1p}\|D^2u\|, 
  \quad u\in H^2(a,b)\cap H^1_0(a,b), \ p=2,\infty.
\end{align}
To prove this we use the identity
\begin{align*}
D(u-I_hu)(x)
=h_{j}^{-1}\int_{x_{j-1}}^{x_{j}} \big(u'(x)-u'(y)\big)\,dy
=h_j^{-1}\int_{x_{j-1}}^{x_{j}} \int_y^xu''(z)\,dz\,dy,
\end{align*}
for $x\in (x_{j-1},x_{j})$ and with $h_j=x_{j}-x_{j-1}$, which yields
\begin{align*}
   |D(u-I_hu)(x)|\le h_{j}^{\frac12}\|D^2u\|_{L_2(x_{j-1},x_{j})} 
   \le h^{\frac12}\|D^2u\|, 
   \quad x\in (x_{j-1},x_{j}).  
\end{align*}
This proves the case $p=\infty$ and for $p=2$ we have 
\begin{align*}
   \|D(u-I_hu)\|^2
   \le\sum_{j=1}^{N}h_{j}^2\|D^2u\|_{L_2(x_{j-1},x_{j})}^2
    \le h^2  \|D^2u\|^2.
\end{align*}

The finite element problem is based on the weak formulation in
\eqref{paper4:eq:5}: find $v_h\in V_h$ such that
\begin{align}
\label{paper4:eq:5b}
 A(v_h,\phi_h)=(f,\phi_h) \quad \forall \phi_h\in V_h,  
\end{align} 
where $A(\cdot,\cdot)$ is defined in \eqref{paper4:eq:4} with the
integral operator computed as in \eqref{paper4:eq:8}. In the following
theorem we prove convergence estimates with explicit constants. 

\begin{theorem} \label{paper4:thm:numerical} Let $v$ be the solution
  of \eqref{paper4:eq:3} as in Theorem
  \ref{paper4:thm:existence}. There is
  $h_0=\sigma/(2^{\frac12}\mu^{\frac12}c_1c_3)$ such that, for $h\le
  h_0$, \eqref{paper4:eq:5b} has a unique solution $v_h\in V_h$ and
  \begin{align}
    \label{paper4:eq:11}
    \|v-v_h\|\le 4c_1^2c_3^2\sigma^{-2} h^2\|f\|, 
  \quad  \|D(v-v_h)\|\le 4c_1c_3\sigma^{-2} h\|f\|.  
  \end{align}
\end{theorem}

\begin{proof}  We adapt an argument from \cite{Schatz}.
  Let $e=v-v_h$ denote the error.  By subtraction of \eqref{paper4:eq:5b} and
  \eqref{paper4:eq:5} with $\phi=\phi_h\in V_h\subset H^1_0(a,b)$ we get 
\begin{align}
\label{paper4:eq:5c}
 A(e,\phi_h)=0 \quad \forall \phi_h\in V_h.  
\end{align} 
Consider the adjoint problem \eqref{paper4:eq:5a} with $g=e$ and
solution $w=(\mathcal{A}^*)^{-1}e$.  With $\phi=e$ this yields
\begin{align*}
  \begin{split}
    \|e\|^2&=A(e,w)=A(e,w-I_hw) \le c_1\|De\|\|D(w-I_hw)\| \\
    &\le c_1\|De\| h\|D^2w\|
    \le c_1c_3h\|De\| \|e\|.  
 \end{split}
\end{align*}
Here we used \eqref{paper4:eq:5c}, \eqref{paper4:eq:12b},
\eqref{paper4:eq:12a}, and \eqref{paper4:eq:6b}.  
We conclude
\begin{align}
  \label{paper4:eq:13}
      \|e\|\le c_1c_3h\|De\|.
\end{align}

In view of \eqref{paper4:eq:5c} we have $A(e,e)=A(e,v-v_h)=A(e,v)$, so that
by \eqref{paper4:eq:12} and \eqref{paper4:eq:13}, 
\begin{align}
 \label{paper4:eq:14}
 \begin{split}
     \tfrac12{\sigma^2} \|De\|^2  
  &\le  A(e,e) +\tfrac12{\mu} \|e\|^2 
   =  A(e,v) +\tfrac12{\mu} \|e\|^2 \\
  &\le c_1\|De\|\|Dv\| +\tfrac12{\mu}c_1^2c_3^2h^2 \|De\|^2 .  
\end{split}
\end{align}
Hence, for $h\le h_0$ sufficiently small ($h_0^2=\sigma^2/(2\mu
c_1^2c_3^2)$), we have
\begin{align*}
  \|De\|\le c_9\|Dv\|,  \quad c_9=4c_1\sigma^{-2}.  
\end{align*}
Now if $f=0$ in \eqref{paper4:eq:5} and \eqref{paper4:eq:5b}, then
$v=0$ by uniqueness, and hence $e=0$, so that $v_h=0$.  This means that
we have uniqueness for the finite element problem
\eqref{paper4:eq:5b}.  But this is an equation in a finite dimensional
space so existence also follows.  Therefore, \eqref{paper4:eq:5b} has
a unique solution for all $f\in L_2(a,b)$ if $h\le h_0$.

In order to prove the error estimate \eqref{paper4:eq:11} we return to
\eqref{paper4:eq:14} but use $A(e,e)=A(e,v-v_h)=A(e,v-I_hv)$ instead:
\begin{align*}
 \begin{split}
     \tfrac12{\sigma^2} \|De\|^2  
  &\le  A(e,e) +\tfrac12{\mu} \|e\|^2
  =  A(e,v-I_hv) +\tfrac12{\mu} \|e\|^2 \\
  &\le c_1\|De\|\|D(v-I_hv)\| +\tfrac12{\mu}c_1^2c_3^2h^2 \|De\|^2,
\end{split}
\end{align*}
and conclude, for $h\le h_0$, 
\begin{align*}
   \|De\|\le c_9\|D(v-I_hv)\|, \quad c_9=4c_1\sigma^{-2}.  
\end{align*}
Hence,  by \eqref{paper4:eq:12a}, \eqref{paper4:eq:26}, and \eqref{paper4:eq:13}, 
\begin{align*}
   \|De\|&\le c_9h \|D^2v\|\le c_9c_3h \|f\|= 4c_1c_3\sigma^{-2}h \|f\|,  \\
   \|e\|&\le c_1c_3h \|De\|\le 4c_1^2c_3^2\sigma^{-2}h^2 \|f\|,
\end{align*}
which is \eqref{paper4:eq:11}.  
\end{proof}

We finish by proving the pointwise convergence of the derivative.

\begin{corollary}  \label{paper4:coro1}
Assume that each finite element mesh $\mathcal{T}_h$ is uniform, that
is, $x_j-x_{j-1}=h$ for $j=1,\dots,N$.  Then, for $h\le h_0$ as in
Theorem \ref{paper4:thm:numerical}, we have 
\begin{align*}
 |v'(b)-v_h'(b)| \le c_{10}h^{\frac12}\|f\| , 
\quad c_{10}=2+4c_1c_3\sigma^{-2}.
\end{align*}
\end{corollary}

\begin{proof}
We use the inverse inequality
\begin{align*}
\|D\phi_h\|_{L_\infty(a,b)} \le  h^{-\frac12}\|D\phi_h\|, \quad
\phi_h\in V_h.
\end{align*}
To prove this we note that
\begin{align*}
D\phi_h(x)
=h^{-1}\int_{x_{j-1}}^{x_{j}} D\phi_h(y)\,dy, 
\quad x\in(x_{j-1},x_{j}),\ h=x_{j}-x_{j-1}, 
\end{align*}
which yields
\begin{align*}
   |D\phi_h(x)|\le h^{-\frac12}\|D\phi_h\|_{L_2(x_{j-1},x_{j})} 
   \le h^{-\frac12}\|D\phi_h\|, 
   \quad x\in (x_{j-1},x_{j}).  
\end{align*}
Hence, by \eqref{paper4:eq:12a} and \eqref{paper4:eq:11}, 
\begin{align*}
\begin{split}
\|De\|_{L_\infty(a,b)} 
 &\le   \|D(v-I_hv)\|_{L_\infty(a,b)}
+ \|D(I_hv-v_h)\|_{L_\infty(a,b)} 
 \\&\le   \|D(v-I_hv)\|_{L_\infty(a,b)}
+ h^{-\frac12}\|D(I_hv-v_h)\|
 \\&\le   \|D(v-I_hv)\|_{L_\infty(a,b)}
+ h^{-\frac12}\|D(I_hv-v)\|
+ h^{-\frac12}\|D(v-v_h)\|
\\& \le 2h^{\frac12}\|D^2v\|+ h^{-\frac12}\|D(v-v_h)\|  
\le (2+4c_1c_3\sigma^{-2})h^{\frac12}\|f\| .
\end{split}
\end{align*}
Therefore 
\begin{align*}
|v'(b)-v_h'(b)| \le  (2+4c_1c_3\sigma^{-2})h^{\frac12}\|f\| .
\end{align*} 
\end{proof}
\noindent In particular, with $f(x)=-\mu x$, Corollary~\ref{paper4:coro1} gives
\begin{align}\label{paper4:final}
|v'(b)-v_h'(b)| \le c_{11}h^{\frac12},\quad c_{11}=c_{10}\mu\sqrt{\frac{b^3-a^3}{3}}.
\end{align}

Given numerical values for the parameters $a,b,\sigma,\mu,\lambda$ we
may now compute numerical values for $h_0$ and $c_{11}$.
Alternatively, we may conclude that there are uniform bounds $h_0\ge
\hat{h}_0$, $c_{11}\le \hat{c}_{11}$ for $b\in[b_1,b_2]$ and with the
other parameters fixed.  

\subsection{The free boundary value problem}  \label{paper4:subsec:3} 

We use uniform meshes $\mathcal{T}_h$ with
\[
x_j-x_{j-1}=h=\frac{b-a}{N}, \quad j=1,\dots,N.
\]
Since we want to vary $b$, we parametrize by $N$ instead of $h$. Let $f(x)=-\mu x$, fix $a<0$ and let $v$, $v_N$ denote the solutions of \eqref{paper4:eq:5} and \eqref{paper4:eq:5b} for $b>a$.   Define the functions 
\begin{align*}
 F(b)=v^{\prime}(b),  \quad
F_N(b)=v_N^{\prime}(b).  
\end{align*} 
From \eqref{paper4:final}, we get for $a<b_1<b_2$ 
\begin{align}\label{paper4:uni_b_est}
\begin{split}
  \| F-F_N\|_{L_\infty(b_1,b_2)}&\le \hat c_{12}N^{-\frac12}, \quad N\ge \hat N_0, \\ 
\hat c_{12}&=\hat c_{11}(b_2-a)^{\frac12}, \quad \hat N_0=\frac{b_2-a}{\hat h_0}.  
\end{split}
\end{align}
By writing down the matrix equation for solving the finite element problem \eqref{paper4:eq:5b}, it is easy to see that, for fixed $N$, the function $b\mapsto F_N(b)$ is continuous on $(a,\infty)$. From \eqref{paper4:uni_b_est} we conclude that $b\mapsto F(b)$ is also continuous on $(a,\infty)$. Moreover, by a direct consequence of the strong maximum principle and the Hopf boundary point principle for our equation (see \cite[Theorem 3.1.4-3.1.5]{Garroni_Menaldi}), we get the following:
\begin{lemma}
If $a<b\leq 0$, then $F(b)<0$. In particular, if $(u,b)$ is a solution to the free boundary problem \eqref{paper4:FBP}, then $b>0$. 
\end{lemma}
We believe that there exists a unique $b>0$ such that
$F(b)=0$. We are not able to provide a rigorous proof of this, but
numerical simulations present strong evidence in the following
way. Assign numerical values to the parameters $a,\sigma,\mu,\lambda$
and fix a jump density $\varphi$. In all our computations, we took
$\varphi$ to be the truncated normal distribution with mean zero, variance
$\gamma>0$ and support $[-J,J]$, i.e.\
\[
\varphi(y)=\begin{cases}
\begin{aligned}
&\frac{e^{-\frac{y^2}{2\gamma^2}}}{\gamma\sqrt{2\pi}\left(2\Phi(J/\gamma)-1\right)}  \quad &&\text{if} \quad -J<y<J, \\
&0 \quad &&\text{otherwise,}
\end{aligned}
\end{cases}
\] 
where
\[
\Phi(x)=\frac{1}{\sqrt{2\pi}}\int_{-\infty}^x e^{-\frac{y^2}{2}}\,dy, \quad x\in\RR.
\]
From computations of the boundary value problem
\eqref{paper4:eq:5b} (see Figures~\ref{paper4:u_prime_b} and \ref{paper4:FBP_sol}), we can find $0\leq b_1<b_2$ and $\tilde N\geq \hat N_0$ such that
\[
F_{\tilde N}(b_1)\leq -\frac{1}{2}, \quad F_{\tilde N}(b_2)\geq \frac{1}{2}, \quad \text{and} \quad\hat c_{12}\tilde N^{-\frac{1}{2}}<\frac{1}{4}.
\]
(The $1/2$ and $1/4$ may vary if we change the parameters.) From \eqref{paper4:uni_b_est}, we can then conclude that
\begin{align*}
\begin{aligned}
F(b_1)&<0,  &&\;\,\, F(b_2)>0, \\
F_N(b_1)&<0,  &&F_N(b_2)>0 \quad \text{for all $N\geq \tilde N$.}
\end{aligned}
\end{align*}
Hence, there exists $b\in (b_1,b_2)$ such that $F(b)=0$ and for each $N\geq \tilde N$ there exists $b_N\in(b_1,b_2)$ such that $F_N(b_N)=0$. Moreover, \eqref{paper4:uni_b_est} gives us that
\[
\lim_{N\to\infty}F(b_N)=0.
\] 
Of course, we cannot conclude that $b$ is unique and $b_N\to b$ as
$N\to\infty$. However, Figure~\ref{paper4:u_prime_b} suggests that $b$
is unique and from computations with increasing $N$, it seems like $b_N$ converges, see Table~\ref{paper4:conv_bN}. 

We now discuss whether the properties $a)$ and $b)$ in the statement of
Theorem~\ref{paper4:ver_thm} hold for a solution $(u,b)$ of
\eqref{paper4:FBP}. We have no rigorous proof, but computational
evidence. The properties $a)$ and $b)$ boil down to
\begin{align}\label{paper4:prop_a}
\lambda \int_a^b v(y)\varphi(y-x)\,dy\leq \mu x,\quad \text{for $x>b$},
\end{align}
and $v\geq 0$ respectively, where $(v,b)$ solves
\eqref{paper4:FBP2}. We believe that $v\geq 0$ holds for all values of
the parameters, but computations suggests that \eqref{paper4:prop_a}
may fail for certain parameter values, typically when $\sigma$ is
small and $\lambda$ is three or four times larger than $\mu$. See Figures~\ref{paper4:cond_gen2} and \ref{paper4:cond_gen1},
  where we check \eqref{paper4:prop_a} for $(v_N,b_N)$ instead of $(v,b)$. 
\begin{figure}[!h]
\begin{center}
\scalebox{\scalefactor}{% This file is generated by the MATLAB m-file laprint.m. It can be included
% into LaTeX documents using the packages graphicx, color and psfrag.
% It is accompanied by a postscript file. A sample LaTeX file is:
%    \documentclass{article}\usepackage{graphicx,color,psfrag}
%    \begin{document}\input{u_prime_b_T}\end{document}
% See http://www.mathworks.de/matlabcentral/fileexchange/loadFile.do?objectId=4638
% for recent versions of laprint.m.
%
% created by:           LaPrint version 3.16 (13.9.2004)
% created on:           15-Apr-2010 04:02:21
% eps bounding box:     15 cm x 11.25 cm
% comment:              
%
\begin{psfrags}%
\psfragscanon%
%
% text strings:
\psfrag{s03}[t][t]{\color[rgb]{0,0,0}\setlength{\tabcolsep}{0pt}\begin{tabular}{c}$b$\end{tabular}}%
\psfrag{s04}[b][b]{\color[rgb]{0,0,0}\setlength{\tabcolsep}{0pt}\begin{tabular}{c}$F_N(b)$\end{tabular}}%
%
% xticklabels:
\psfrag{x01}[t][t]{0}%
\psfrag{x02}[t][t]{0.02}%
\psfrag{x03}[t][t]{0.04}%
\psfrag{x04}[t][t]{0.06}%
\psfrag{x05}[t][t]{0.08}%
\psfrag{x06}[t][t]{0.1}%
%
% yticklabels:
\psfrag{v01}[r][r]{-1}%
\psfrag{v02}[r][r]{-0.5}%
\psfrag{v03}[r][r]{0}%
\psfrag{v04}[r][r]{0.5}%
\psfrag{v05}[r][r]{1}%
\psfrag{v06}[r][r]{1.5}%
\psfrag{v07}[r][r]{2}%
%
% Figure:
\resizebox{12cm}{!}{\includegraphics{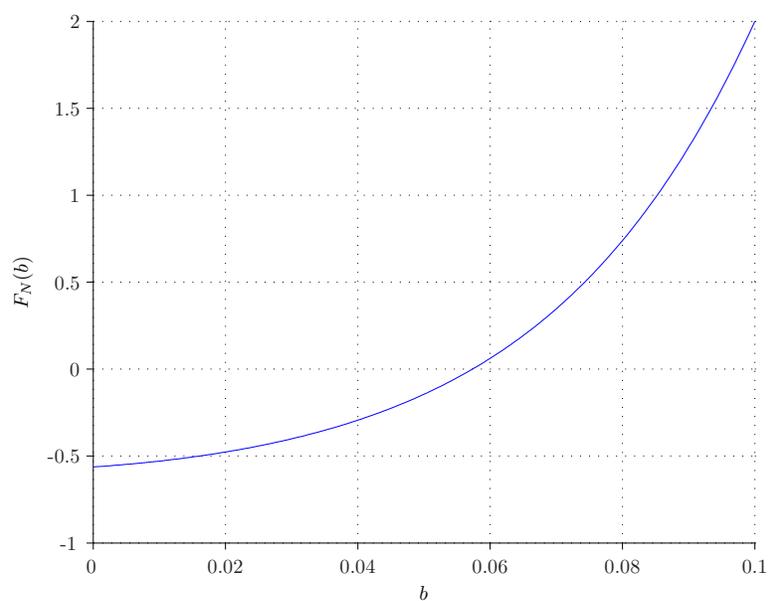}}%
\end{psfrags}%
%
% End u_prime_b_T.tex
}
\caption{The function $F_N$ when $a=-0.1$, $\lambda=10$, $\sigma=0.2$, $\mu=\frac{\sigma^2}{0.005}$, $\gamma=0.02$ and $J=0.05$.}
\label{paper4:u_prime_b}
\end{center}
\end{figure}
% \begin{table}[!h]
% \begin{center}
% \begin{tabular}{c|c}
% $N$ & $b_N$ \\
% \hline
% 1000 & 1.121999 \\
% \hline
% 2000 & 1.121217 \\
% \hline
% 3000 & 1.120957 \\
% \hline
% 4000 & 1.120826 \\
% \hline
% 5000 & 1.120748 \\
% \hline
% 6000 & 1.120696 \\
% \hline
% 7000 & 1.120659 \\
% \hline
% 8000 & 1.120631
% \end{tabular} 
% \caption{$a=-2$, $\lambda=1$, $\mu=10$, $\sigma=5$, $\gamma=1$ and $J=4$.}
% \label{paper4:conv_bN}
% \end{center}
% \end{table}
\begin{table}[!h]
\begin{center}
\begin{tabular}{c|c}
$N$ & $b_N$ \\
\hline
2000 & 0.0572939 \\
\hline
4000 & 0.0572743 \\
\hline
6000 & 0.0572678 \\
\hline
8000 & 0.0572653
\end{tabular} 
\caption{$a=-0.1$, $\lambda=10$, $\sigma=0.2$, $\mu=\frac{\sigma^2}{0.005}$, $\gamma=0.02$ and $J=0.05$.}
\label{paper4:conv_bN}
\end{center}
\end{table}
\begin{figure}[!h]
\begin{center}
\scalebox{\scalefactor}{% This file is generated by the MATLAB m-file laprint.m. It can be included
% into LaTeX documents using the packages graphicx, color and psfrag.
% It is accompanied by a postscript file. A sample LaTeX file is:
%    \documentclass{article}\usepackage{graphicx,color,psfrag}
%    \begin{document}\input{FBP_sol_T}\end{document}
% See http://www.mathworks.de/matlabcentral/fileexchange/loadFile.do?objectId=4638
% for recent versions of laprint.m.
%
% created by:           LaPrint version 3.16 (13.9.2004)
% created on:           15-Apr-2010 04:02:20
% eps bounding box:     15 cm x 11.25 cm
% comment:              
%
\begin{psfrags}%
\psfragscanon%
%
% text strings:
\psfrag{s03}[t][t]{\color[rgb]{0,0,0}\setlength{\tabcolsep}{0pt}\begin{tabular}{c}$x$\end{tabular}}%
\psfrag{s04}[b][b]{\color[rgb]{0,0,0}\setlength{\tabcolsep}{0pt}\begin{tabular}{c}$v_N (x)$\end{tabular}}%
\psfrag{s05}[l][l]{\color[rgb]{0,0,0}\setlength{\tabcolsep}{0pt}\begin{tabular}{l}$b_N=0.0573$\end{tabular}}%
%
% xticklabels:
\psfrag{x01}[t][t]{-0.1}%
\psfrag{x02}[t][t]{-0.05}%
\psfrag{x03}[t][t]{0}%
\psfrag{x04}[t][t]{0.05}%
%
% yticklabels:
\psfrag{v01}[r][r]{0}%
\psfrag{v02}[r][r]{0.01}%
\psfrag{v03}[r][r]{0.02}%
\psfrag{v04}[r][r]{0.03}%
\psfrag{v05}[r][r]{0.04}%
\psfrag{v06}[r][r]{0.05}%
%
% Figure:
\resizebox{12cm}{!}{\includegraphics{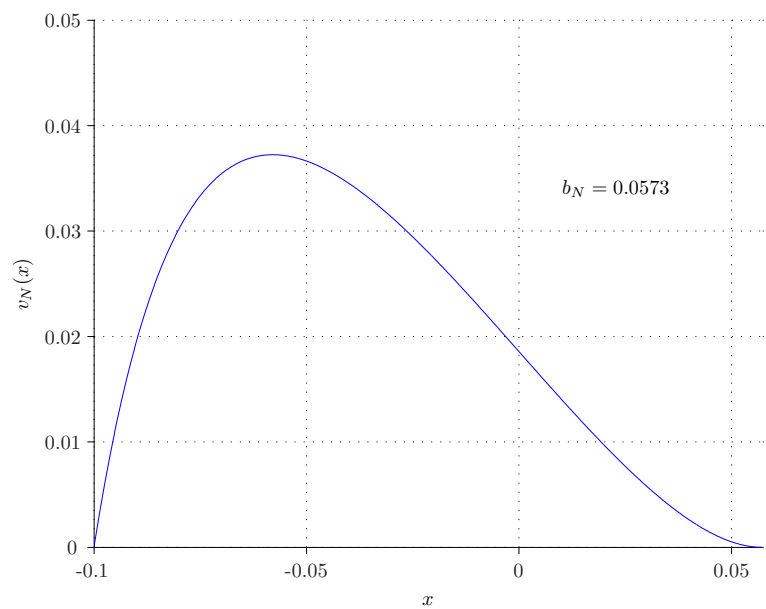}}%
\end{psfrags}%
%
% End FBP_sol_T.tex
}
\caption{The solution $(v_N,b_N)$ when $a=-0.1$, $\lambda=10$, $\sigma=0.2$, $\mu=\frac{\sigma^2}{0.005}$, $\gamma=0.02$ and $J=0.05$.}
\label{paper4:FBP_sol}
\end{center}
\end{figure}
\begin{figure}[!h]
\begin{center}
\scalebox{\scalefactor}{% This file is generated by the MATLAB m-file laprint.m. It can be included
% into LaTeX documents using the packages graphicx, color and psfrag.
% It is accompanied by a postscript file. A sample LaTeX file is:
%    \documentclass{article}\usepackage{graphicx,color,psfrag}
%    \begin{document}\input{cond_gen2_T}\end{document}
% See http://www.mathworks.de/matlabcentral/fileexchange/loadFile.do?objectId=4638
% for recent versions of laprint.m.
%
% created by:           LaPrint version 3.16 (13.9.2004)
% created on:           15-Apr-2010 04:02:22
% eps bounding box:     15 cm x 11.6043 cm
% comment:              
%
\begin{psfrags}%
\psfragscanon%
%
% text strings:
\psfrag{s01}[t][t]{\color[rgb]{0,0,0}\setlength{\tabcolsep}{0pt}\begin{tabular}{c}$x$\end{tabular}}%
\psfrag{s02}[l][l]{\color[rgb]{0,0,0}\setlength{\tabcolsep}{0pt}\begin{tabular}{l}$x\mapsto\lambda\int_a^{b_N} v_N (y)\varphi(y-x)\, dy$\end{tabular}}%
\psfrag{s03}[l][l]{\color[rgb]{0,0,0}\setlength{\tabcolsep}{0pt}\begin{tabular}{l}$x\mapsto\mu x$\end{tabular}}%
\psfrag{s04}[l][l]{\color[rgb]{0,0,0}\setlength{\tabcolsep}{0pt}\begin{tabular}{l}$b_N=0.0560$\end{tabular}}%
%
% xticklabels:
\psfrag{x01}[t][t]{0.06}%
\psfrag{x02}[t][t]{0.07}%
\psfrag{x03}[t][t]{0.08}%
\psfrag{x04}[t][t]{0.09}%
\psfrag{x05}[t][t]{0.1}%
%
% yticklabels:
\psfrag{v01}[r][r]{0}%
\psfrag{v02}[r][r]{0.2}%
\psfrag{v03}[r][r]{0.4}%
\psfrag{v04}[r][r]{0.6}%
\psfrag{v05}[r][r]{0.8}%
\psfrag{v06}[r][r]{1}%
\psfrag{v07}[r][r]{1.2}%
%
% Figure:
\resizebox{12cm}{!}{\includegraphics{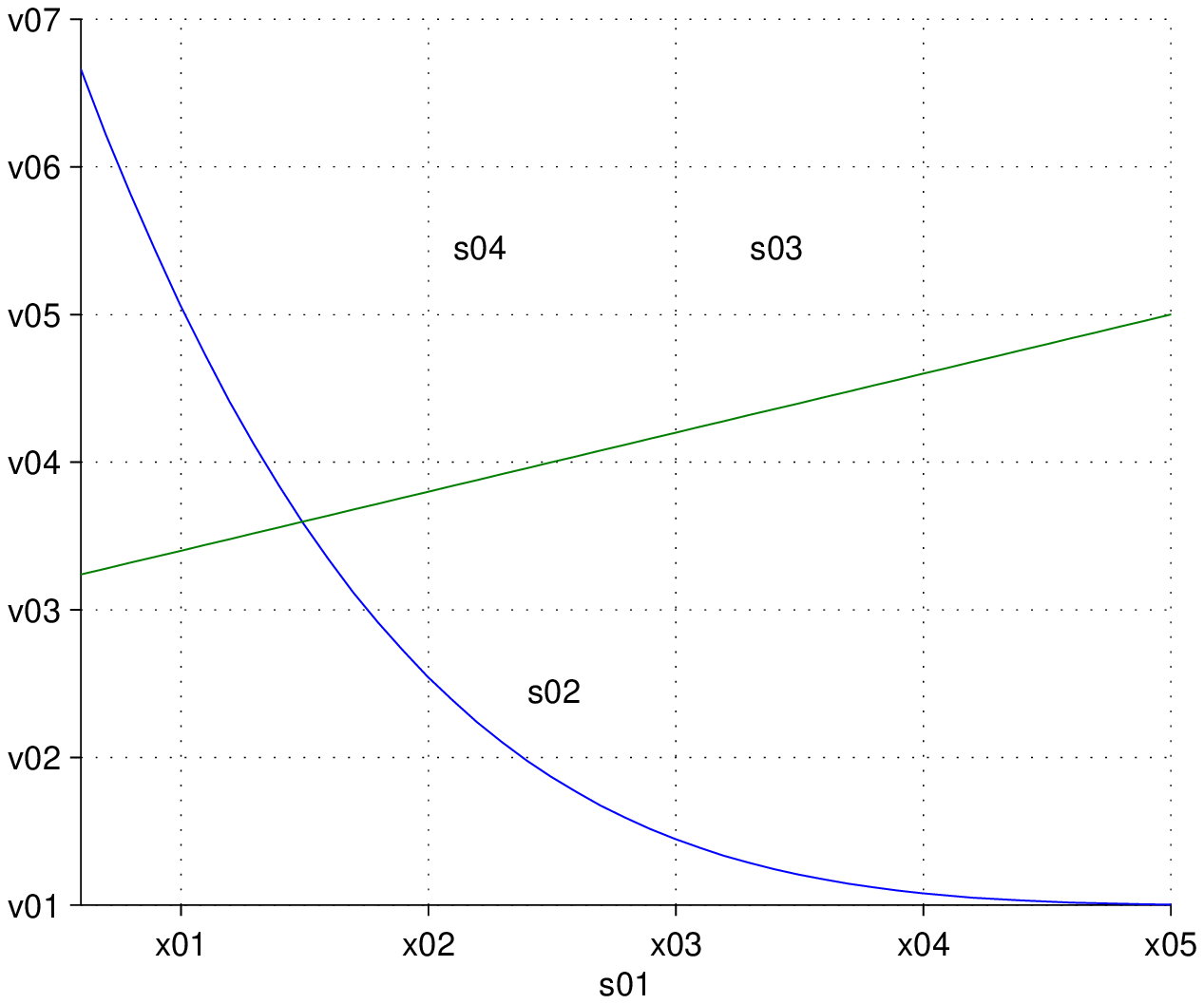}}%
\end{psfrags}%
%
% End cond_gen2_T.tex
}
\caption{A simulation of \eqref{paper4:prop_a} when $a=-0.1$,
  $\lambda=30$, $\sigma=0.2$, $\mu=\frac{\sigma^2}{0.005}$, $\gamma=0.02$ and $J=0.05$. The condition
  fails.}
\label{paper4:cond_gen2}
\end{center}
\end{figure}

\begin{figure}[!h]
\begin{center}
\scalebox{\scalefactor}{% This file is generated by the MATLAB m-file laprint.m. It can be included
% into LaTeX documents using the packages graphicx, color and psfrag.
% It is accompanied by a postscript file. A sample LaTeX file is:
%    \documentclass{article}\usepackage{graphicx,color,psfrag}
%    \begin{document}\input{cond_gen1_T}\end{document}
% See http://www.mathworks.de/matlabcentral/fileexchange/loadFile.do?objectId=4638
% for recent versions of laprint.m.
%
% created by:           LaPrint version 3.16 (13.9.2004)
% created on:           15-Apr-2010 04:02:22
% eps bounding box:     15 cm x 11.25 cm
% comment:              
%
\begin{psfrags}%
\psfragscanon%
%
% text strings:
\psfrag{s04}[t][t]{\color[rgb]{0,0,0}\setlength{\tabcolsep}{0pt}\begin{tabular}{c}$x$\end{tabular}}%
\psfrag{s05}[l][l]{\color[rgb]{0,0,0}\setlength{\tabcolsep}{0pt}\begin{tabular}{l}$x\mapsto\lambda\int_a^{b_N} v_N (y)\varphi(y-x)\, dy$\end{tabular}}%
\psfrag{s06}[l][l]{\color[rgb]{0,0,0}\setlength{\tabcolsep}{0pt}\begin{tabular}{l}$x\mapsto\mu x$\end{tabular}}%
\psfrag{s07}[l][l]{\color[rgb]{0,0,0}\setlength{\tabcolsep}{0pt}\begin{tabular}{l}$b_N=0.0573$\end{tabular}}%
%
% xticklabels:
\psfrag{x01}[t][t]{0.06}%
\psfrag{x02}[t][t]{0.065}%
\psfrag{x03}[t][t]{0.07}%
\psfrag{x04}[t][t]{0.075}%
\psfrag{x05}[t][t]{0.08}%
%
% yticklabels:
\psfrag{v01}[r][r]{0}%
\psfrag{v02}[r][r]{0.2}%
\psfrag{v03}[r][r]{0.4}%
\psfrag{v04}[r][r]{0.6}%
\psfrag{v05}[r][r]{0.8}%
%
% Figure:
\resizebox{12cm}{!}{\includegraphics{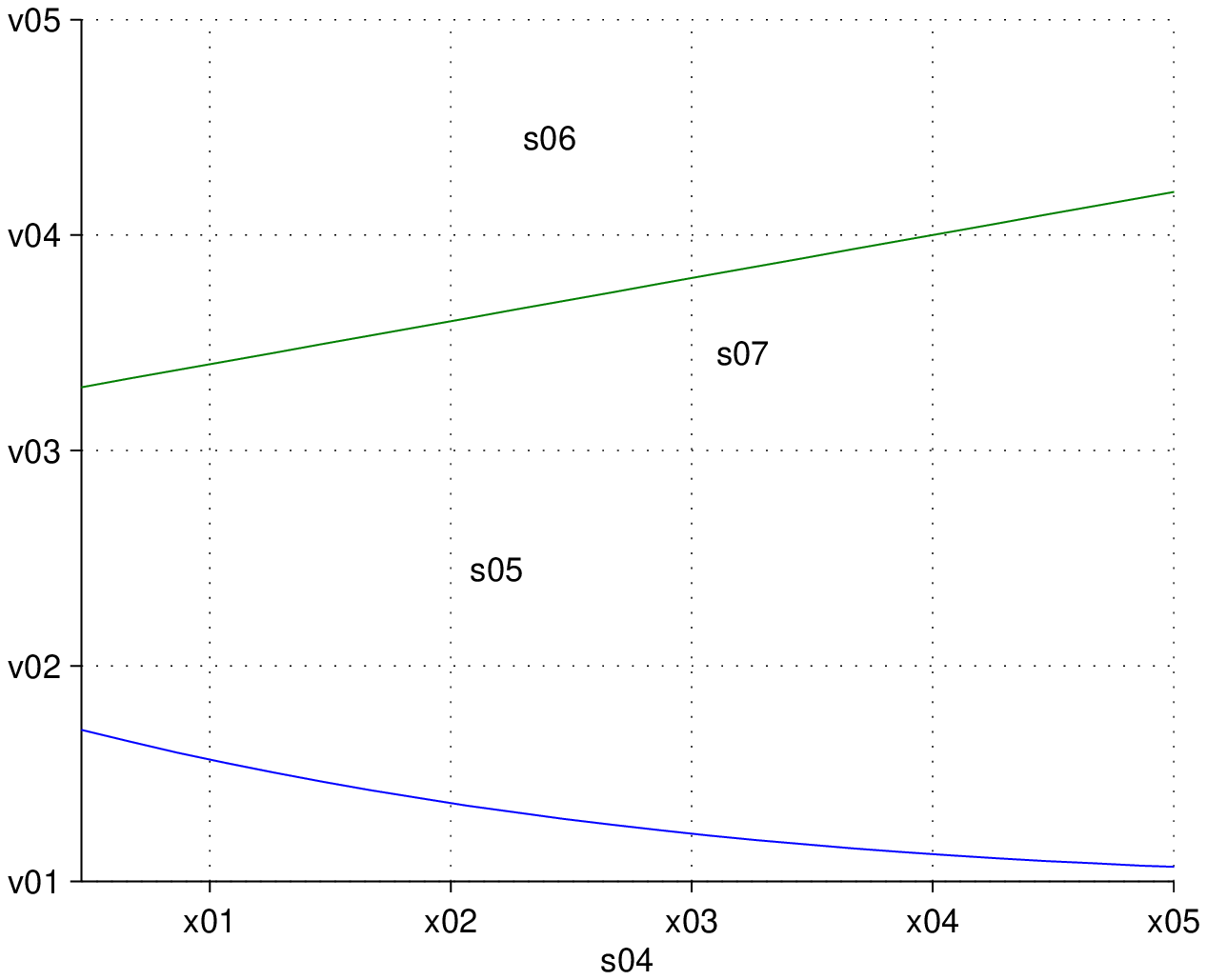}}%
\end{psfrags}%
%
% End cond_gen1_T.tex
}
\caption{A simulation of \eqref{paper4:prop_a} when $a=-0.1$,
  $\lambda=10$, $\sigma=0.2$, $\mu=\frac{\sigma^2}{0.005}$, $\gamma=0.02$ and $J=0.05$. The condition
  holds.}
\label{paper4:cond_gen1}
\end{center}
\end{figure}

%\section{Conjectures}

%\section*{Acknowledgement}

%--------------------------------------------------------------------------------------

\bibliographystyle{amsplain}
\bibliography{/home/warfheimer/Backup_remote/Forskning/Paper1/references}
%\bibliography{/chalmers/users/warfheim/Forskning/Paper1/references}

%\begin{thebibliography}{10}

%\end{thebibliography}

%M-! dvips -o paper4.ps paper4.dvi

\end{document}